\documentclass[11pt, a4paper]{article}

\usepackage[utf8]{inputenc}   
\usepackage[T1]{fontenc}      
\usepackage[english]{babel}   
\usepackage{amsmath, amssymb, amsthm, mathtools}
\usepackage{bm}              
\usepackage{textcomp}        
\usepackage{enumitem}        
\usepackage{geometry}       
\usepackage{booktabs,tabularx,array}
\usepackage{graphicx}          
\usepackage{tikz-cd}		   
\usepackage{xcolor}            
\DeclareGraphicsExtensions{.pdf,.eps,.jpg,.png}
\graphicspath{{figures/}{figure/}{pictures/}{picture/}{pic/}{pics/}{image/}{images/}}

\usepackage[colorlinks=true, allcolors=blue]{hyperref} 
\usepackage[capitalize]{cleveref}
\usepackage[numbers, sort&compress]{natbib}

\makeatletter

\newcommand{\Rmnum}[1]{\expandafter\@slowromancap\romannumeral #1@}
\makeatother

\newcommand{\E}{\mathbb E}
\newcommand{\Pp}{\mathbb P}
\newcommand{\Var}{\operatorname{Var}}
\newcommand{\Cov}{\operatorname{Cov}}
\newcommand{\LambdaL}{\Lambda_L}
\newcommand{\BL}{B_L}
\newcommand{\EL}{E_L}
\newcommand{\dd}{\,\mathrm d}
\newcommand{\abs}[1]{\lvert #1\rvert}
\newcommand{\norm}[1]{\lVert #1\rVert}

\theoremstyle{plain} 
\newtheorem{theorem}{Theorem}
\newtheorem{lemma}[theorem]{Lemma}

\theoremstyle{definition}

\theoremstyle{remark}

\renewcommand{\[}{\begin{equation}}
\renewcommand{\]}{\end{equation}}

\title{\Large\bf Boundary Free Energies in Disordered Ising Models}
\author{Hexiang Wang \and Keheng Zhu \and Mauris Chueng}
\newcommand{\Addresses}{{
		\bigskip
		\footnotesize
		
		\textsc{Hexiang Wang}, \textsc{School of Mathematical Sciences, Nankai University, Tianjin, 300071, China}\par\nopagebreak
		\texttt{Kui6539@outlook.com}
		\medskip
		
		\textsc{Keheng Zhu}, \textsc{Academy for Multidisciplinary Studies, School of Mathematics Sciences, Capital Normal
			University, Beijing, 100048, China}\par\nopagebreak
		\texttt{hexistartop@gmail.com}
		\medskip
		
		\textsc{Mauris Chueng}, \textsc{School of Statistics and Data Science, Jilin University of Finance and Economics, Changchun, 130117, China}\par\nopagebreak
		\texttt{maurischueng@gmail.com}
		\medskip
}}
\date{}

\begin{document}
	
	\maketitle
	
\begin{abstract}
The boundary correction to the free energy of a disordered Ising model depends
on the boundary condition, the normalization, and the finite-volume sequence.
We give a one-dimensional i.i.d. counterexample showing that the surface
free-energy density need not be independent of the van Hove sequence and that
the corresponding random correction need not converge in probability. For
bounded couplings in the Dobrushin uniqueness regime on cubic boxes in
dimensions $d\geq2$, we prove convergence of the normalized free-to-fixed
boundary free-energy difference in expectation, almost surely, and in every
$L^p$, $1\leq p<\infty$. For Gaussian boundary couplings, we derive an
exact finite-volume interpolation identity and explain why it does not by
itself imply a low-temperature surface limit. For a seam-flip free-energy
difference $D_L$ across a set $S_L$ of independent symmetric bonds of variance
$v$, we prove $\Var(D_L)\leq4v\abs{S_L}$; one-dimensional examples show that
symmetry and finite moments alone do not determine a stiffness exponent, and
the low-temperature Gaussian Edwards--Anderson problem remains open.
\end{abstract}

\section{Introduction}

Consider an experimentalist who measures the free energy of the same
disordered magnet twice. In the first experiment the boundary is left free;
in the second, the sample is placed in contact with a frozen exterior spin
configuration. For the bounded couplings considered in our surface-limit
theorem, the two free energies per site become indistinguishable as the box
grows. After the difference is multiplied by the linear size, however, a
boundary contribution may remain, and it is natural to ask whether this
contribution approaches a deterministic constant or records the fluctuations
of a macroscopic interface.

This question has several inequivalent mathematical formulations. Its answer
depends on the dimension, the boundary condition, the disorder law, the shape
of the finite volumes, and whether the disorder is averaged before taking the
limit. We therefore begin by specifying the model and by separating a
free-to-fixed boundary free-energy difference from a periodic/antiperiodic, or
seam-flip, domain-wall free energy.

Our first conclusion is negative: the hypotheses in the question do not imply
a finite-volume-sequence-independent surface limit. In \cref{thm:impossibility} we
give an i.i.d. one-dimensional example with two distinct van Hove
subsequential limits, and we show that even regular intervals need not give a
sample limit in probability. Our positive result, \cref{thm:surface}, proves a
deterministic surface free-energy density for bounded disorder in the
Dobrushin uniqueness regime on cubic boxes. The proof identifies the limit
with a half-space boundary response and establishes almost-sure and $L^p$
convergence in dimensions $d\geq2$. For rectangular boxes whose minimum side
tends to infinity, expectation and $L^p$ convergence remain valid; the
almost-sure conclusion requires the summability condition in the theorem.

For Gaussian boundary disorder, \cref{sec:interpolation} derives an exact
finite-volume interpolation formula. The interpolation proposed in the
question does not start from the free-boundary Hamiltonian; the
endpoint-compatible interpolation does, but its boundary-overlap identity contains no mechanism
that forces a low-temperature thermodynamic limit. Finally,
\cref{thm:seam} gives a temperature-uniform variance estimate for seam-flip
free-energy differences. This estimate bounds a root-mean-square stiffness
exponent but does not establish the existence or value of such an exponent;
see \cite{FisherHuse1988,NewmanStein2003} for the physical and mathematical
background of stiffness scaling in short-range spin glasses.
The unresolved Gaussian low-temperature problem is stated precisely in the
final section.

\section{Model and notation}
\label{sec:model}

\subsection{Geometry, spins, and boundary edges}

Fix $d\geq1$. For the regular-box results let
\[
  \LambdaL=\{1,\ldots,L\}^d\subset\mathbb Z^d,
  \qquad \abs{\LambdaL}=L^d.
\]
The internal nearest-neighbor edges and the oriented boundary edges are
\begin{align*}
  \EL&=\bigl\{\{x,z\}:x,z\in\LambdaL,\ \abs{x-z}_1=1\bigr\},\\
  \BL&=\bigl\{(x,y):x\in\LambdaL,\ y\notin\LambdaL,
                    \ \abs{x-y}_1=1\bigr\}.
\end{align*}
For cubes,
\begin{equation}\label{eq:counts}
  \abs{\EL}=d(L-1)L^{d-1},
  \qquad b_L:=\abs{\BL}=2dL^{d-1}.
\end{equation}
For a finite set $\Lambda\subset\mathbb Z^d$, let $B(\Lambda)$ denote its set
  of oriented nearest-neighbor boundary edges, defined as above. A sequence
  $(\Lambda_n)$ is called a van Hove sequence if, for every fixed $r\geq1$,
  \begin{equation}\label{eq:van-hove}
    \frac{\abs{\{x\in\Lambda_n:
      \operatorname{dist}(x,\Lambda_n^c)\leq r\}}}{\abs{\Lambda_n}}
    \longrightarrow0.
  \end{equation}
This is the standard van Hove condition used in thermodynamic-limit arguments;
see \cite{Georgii2011,Ruelle1969}.

\begin{figure}[htbp]
	\centering
	\begin{minipage}{0.48\textwidth}
		\centering
		\includegraphics[width=\textwidth]{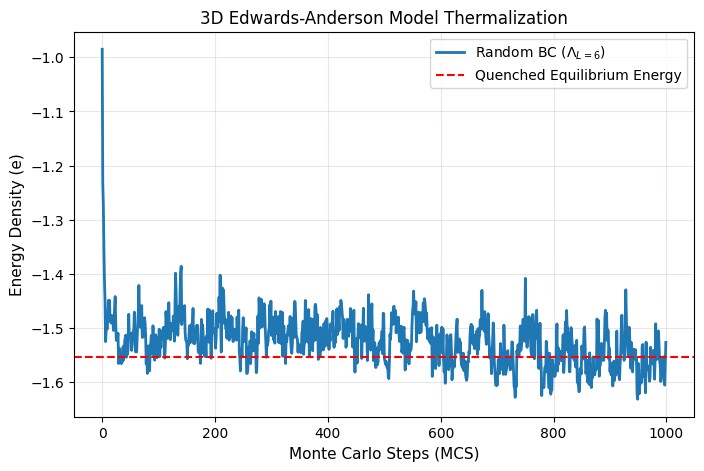}
		\caption{Schematic preview of the half-space lattice geometry $\mathbb{H}$, illustrating the introductory configuration of boundary couplings and the finite-volume box approximation layout.}
		\label{fig:half_space_preview}
	\end{minipage}
	\hfill
	\begin{minipage}{0.48\textwidth}
		\centering
		\includegraphics[width=\textwidth]{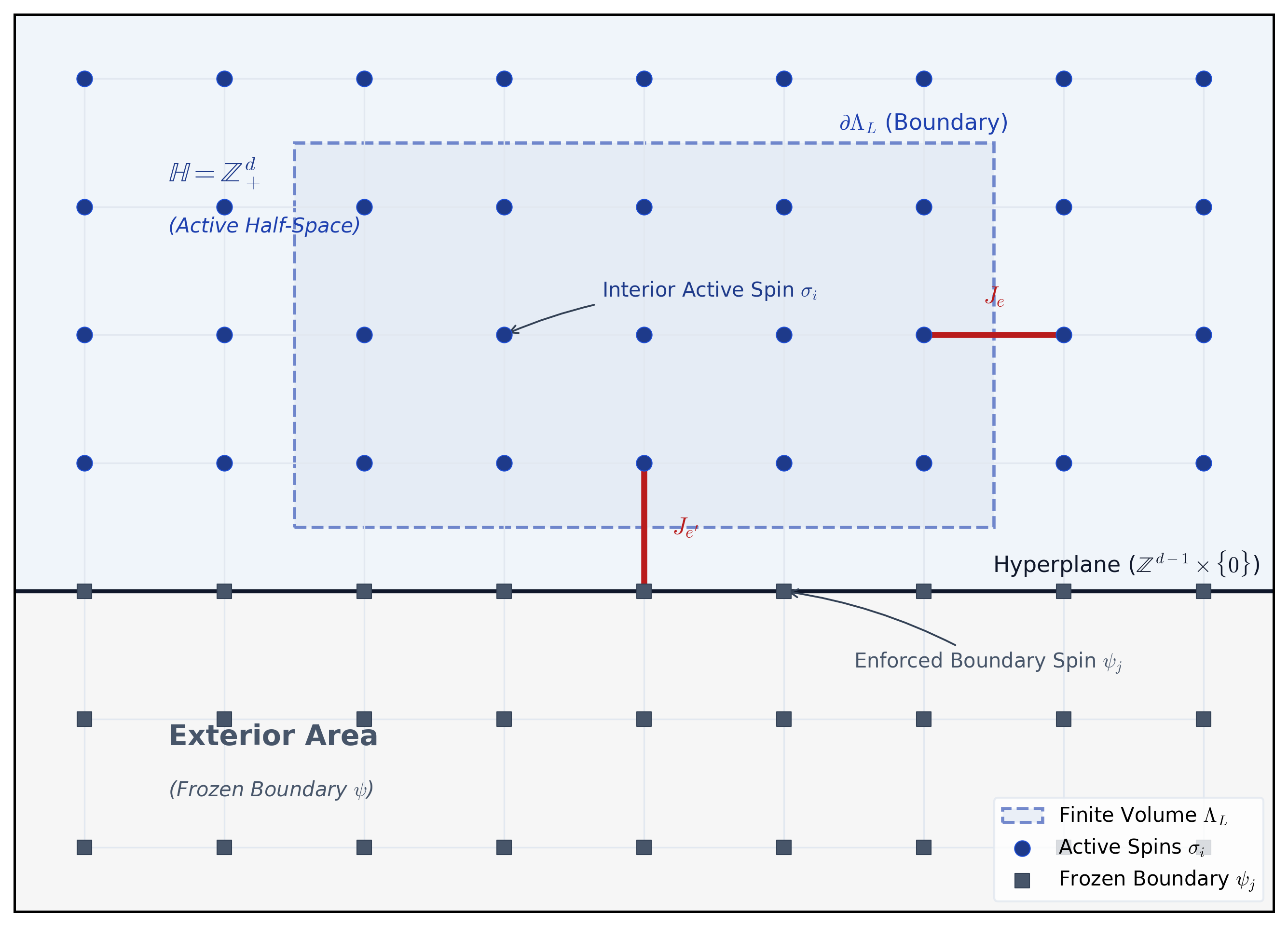}
		\caption{High-resolution numerical rendering of the half-space coupling graph, detailing the boundary layer interactions and spin configurations within the uniform Dobrushin uniqueness regime.}
		\label{fig:half_space_highres}
	\end{minipage}
\end{figure}

Associate an Ising spin $\sigma_x\in\{-1,+1\}$ with each $x\in\LambdaL$.
An exterior spin pattern $\psi=(\psi_y)\in\{-1,+1\}$ is deterministic.
For $e=(x,y)\in\BL$ put $q_e(\sigma)=\sigma_x\psi_y$.

Let $J_a$ be the internal coupling on $a=\{x,z\}\in\EL$ and
$K_e$ the boundary coupling on $e\in\BL$. Unless explicitly stated otherwise,
all couplings are independent. The free Hamiltonian and the Hamiltonian with a
fixed exterior-spin boundary condition are
\begin{align}
  H_L^{\mathrm{free}}(\sigma;J)
    &=-\sum_{a=\{x,z\}\in\EL}J_a\sigma_x\sigma_z,
       \label{eq:Hfree}\\
  H_L^{\mathrm{fix}}(\sigma;J,K,\psi)
    &=H_L^{\mathrm{free}}(\sigma;J)
      -\sum_{e\in\BL}K_e q_e(\sigma).
       \label{eq:Hfix}
\end{align}
Thus the spins outside $\LambdaL$ are prescribed, whereas the spins in
$\LambdaL$ remain unconstrained. This standard fixed boundary condition is
different from hard conditioning of spins inside $\LambdaL$, and it is also
different from periodic or antiperiodic boundary conditions.

For fixed $\beta\in(0,\infty)$ and
$b\in\{\mathrm{free},\mathrm{fix}\}$,
\[
  Z_L^b=\sum_{\sigma\in\{-1,+1\}^{\LambdaL}}
             e^{-\beta H_L^b(\sigma)},
  \quad F_L^b=-\frac1\beta\log Z_L^b,
  \quad f_L^b=\frac{F_L^b}{\abs{\LambdaL}}.
\]
The sample boundary correction and its quenched mean are
\begin{equation}\label{eq:Delta-def}
  \Delta F_L(J,K)=F_L^{\mathrm{fix}}(J,K)-F_L^{\mathrm{free}}(J).
\end{equation}
Thus $\Delta F_L$ is a random free-energy difference, whereas
$\E\Delta F_L$ is its disorder average. We use ``quenched free energy'' for
$\E F_L=-\beta^{-1}\E\log Z_L$ and ``quenched log-partition function'' for
$\E\log Z_L$.

For an arbitrary finite $\Lambda\subset\mathbb Z^d$, the notation
$\Delta F_\Lambda$ refers to the analogous free-to-fixed free-energy
difference, with boundary-edge set $B(\Lambda)$.

\subsection{Hypotheses for the three results}

We use the following three sets of assumptions.

\begin{description}[leftmargin=2.8em,labelindent=0em,style=nextline]
\item[(G) Gaussian interpolation identities.]
The boundary-coupling vector $K=(K_e)_{e\in\BL}$ is centered Gaussian with covariance
$\Gamma$ and is independent of the internal disorder. For the i.i.d. formulas,
$\Gamma=\kappa^2 I$. When concentration over all bonds is used, the internal bonds
are also independent $N(0,\kappa^2)$ variables. Whenever an almost-sure claim
compares different volumes, these i.i.d. variables are jointly realized as
restrictions of a single Gaussian edge field on $\mathbb Z^d$.

\item[(D) Dobrushin uniqueness regime.]
A single i.i.d. symmetric coupling field is defined on all nearest-neighbor
edges of $\mathbb Z^d$ and is used to couple all volumes.
Its restrictions to $\EL$ and $\BL$ are denoted by $J$ and $K$, respectively,
and satisfy $\abs{J_a},\abs{K_e}\leq J_*$. Put
\begin{equation}\label{eq:dob-condition}
  \alpha:=2d\tanh(\beta J_*)<1.
\end{equation}
The exterior pattern is all plus. By symmetry of the boundary couplings, the
same disorder-averaged conclusions hold for every deterministic,
disorder-independent exterior pattern. The limit is $L\to\infty$ at fixed
$\beta$ satisfying \eqref{eq:dob-condition}.

\item[(S) Seam-flip comparison.]
Two systems use the same disorder realization and differ by a sign flip on a seam
$S_L$ of independent symmetric bonds. No high-temperature assumption is made.
The bound remains valid at $\beta=\infty$ for ground-state energies.
Almost-sure statements across volumes use a joint realization of all
finite-volume disorders on one probability space.
\end{description}

Regime (D) yields a surface limit for regular boxes. Regime (G) gives exact
identities at every finite $L$ but does not by itself prove a low-temperature
limit. Regime (S) concerns a domain-wall free energy rather than the
free-to-fixed boundary free-energy difference.

\section{Main results}

\begin{theorem}[Counterexamples to model-independent convergence]
\label{thm:impossibility}
The assumptions stated in the question do not determine a
finite-volume-sequence-independent surface free-energy limit. More precisely:
\begin{enumerate}
\item There is a nested van Hove exhaustion of $\mathbb Z$ and a bounded i.i.d.
  symmetric nearest-neighbor disorder law for which
  $\E\Delta F_\Lambda/\abs{B(\Lambda)}$ has two distinct subsequential limits.
\item Along regular one-dimensional intervals the expected surface term does
  converge, but for a nondegenerate coupling magnitude the sample term need
  not converge even in probability.
\item If $f_L$ denotes $\E F_L/\abs{\LambdaL}$, the proposed quantity is
  deterministic at every $L$ and cannot itself have quenched fluctuations.
  If $f_L$ denotes a sample density, it is a different random object.
\item Fixed exterior-spin, hard fixed-spin, free versus periodic, and
  periodic versus antiperiodic comparisons are inequivalent and can have incompatible
  mean and fluctuation behavior.
\item For one-dimensional periodic versus antiperiodic ground-state
  comparisons, i.i.d. Rademacher and symmetric Laplace couplings have
  root-mean-square stiffness exponents $0$ and $-1$, respectively. Thus
  symmetry and finite moments do not determine a universal exponent.
\end{enumerate}
The first two assertions are proved in \cref{sec:counterexamples}; the fifth is
proved in \cref{eq:ring-ground} and the discussion following it.
\end{theorem}

\begin{theorem}[Surface free energy in the Dobrushin uniqueness regime]
\label{thm:surface}
Assume regime (D), $d\geq2$, and the cube sequence $\LambdaL$. There is a
constant $\tau(\beta)$ such that, for every finite $p$,
\begin{equation}\label{eq:surface-conv-main}
  \frac{\E\Delta F_L}{b_L}
  \xrightarrow[L\to\infty]{}\tau(\beta),
  \qquad
  \frac{\Delta F_L}{b_L}
  \xrightarrow[L\to\infty]{\mathrm{a.s.}}\tau(\beta),
  \qquad
  \frac{\Delta F_L}{b_L}
  \xrightarrow[L\to\infty]{L^p}\tau(\beta),
\end{equation}
for every $1\leq p<\infty$.
Let $\mathbb H=\{x\in\mathbb Z^d:x_1\geq1\}$ and let
$\langle\cdot\rangle_{\mathbb H,s}$ be its unique Gibbs state for the fixed
disorder realization, with
  boundary bonds joining the hyperplanes $x_1=1$ and $x_1=0$ multiplied by
  $s\in[0,1]$. Let $\mathbf e_1$ denote
the first coordinate vector and set $e_0=(\mathbf e_1,0)$. Define
\begin{equation}\label{eq:g-half}
  g_{\mathbb H}(s)=
  \E\!\left[K_{e_0}\langle\sigma_{\mathbf e_1}\rangle_{\mathbb H,s}\right].
\end{equation}
Then
\begin{equation}\label{eq:tau-half}
  \tau(\beta)=-\int_0^1g_{\mathbb H}(s)\dd s.
\end{equation}
If the boundary-coupling law has variance $v_K$, then
$-\beta v_K/2\leq\tau(\beta)\leq0$. Moreover, with
\begin{equation}\label{eq:Cstar}
  C_*=(2J_*)^2+
  \frac{32d(2J_*)^2}{(1-\alpha)^2(1-\alpha^2)},
\end{equation}
the centered correction satisfies the explicit surface-order estimate
\begin{equation}\label{eq:surface-mcdiarmid}
  \Pp\!\left(\abs{\Delta F_L-\E\Delta F_L}\geq u\right)
  \leq2\exp\!\left(-\frac{2u^2}{C_*b_L}\right).
\end{equation}
For the density normalization in the question,
\begin{equation}\label{eq:requested-limit}
  L\bigl(\E f_L^{\mathrm{fix}}-\E f_L^{\mathrm{free}}\bigr)
  \longrightarrow 2d\,\tau(\beta).
\end{equation}
For rectangular boxes $\Lambda_{\mathbf L}$ whose minimum side tends to
infinity, put $\ell=\min_iL_i$ and
$b_{\mathbf L}=\abs{B(\Lambda_{\mathbf L})}$, and write
$\Delta F_{\mathbf L}:=\Delta F_{\Lambda_{\mathbf L}}$. Then
\[
  \frac{\E\Delta F_{\mathbf L}}{b_{\mathbf L}}
  \xrightarrow[\ell\to\infty]{}\tau(\beta),
  \qquad
  \frac{\Delta F_{\mathbf L}}{b_{\mathbf L}}
  \xrightarrow[\ell\to\infty]{L^p}\tau(\beta)
\]
for every finite $p$, hence also in probability. Almost-sure convergence along
a rectangle sequence additionally requires
$\sum_n\exp(-c\abs{B(\Lambda_n)})<\infty$ for every $c>0$. No result is claimed
for arbitrary van Hove sequences.
\end{theorem}

\begin{theorem}[Variance bound for seam-flip free-energy differences]
\label{thm:seam}
Let $S_L$ be a set of independent symmetric seam bonds of variance $v$.
Let the twisted system be obtained from the untwisted system by replacing
$y_e$ with $-y_e$ for every $e\in S_L$, with all other disorder held fixed.
For the same-disorder free-energy difference
$D_L=F_L^{\mathrm{tw}}-F_L^{\mathrm{untw}}$,
\begin{equation}\label{eq:seam-summary}
  \E[D_L\mid\text{non-seam disorder}]=0,
  \quad \E D_L^2\leq4v\abs{S_L},
  \quad \abs{D_L}\leq2\sum_{e\in S_L}\abs{y_e}.
\end{equation}
If the seam variables are Gaussian, then
\begin{equation}\label{eq:seam-tail-summary}
  \Pp(\abs{D_L}\geq u)
  \leq2\exp\!\left(-\frac{u^2}{8v\abs{S_L}}\right).
\end{equation}
Consequently, if $\abs{S_L}=c_S L^{d-1}+o(L^{d-1})$ for some $c_S>0$ and
$(\E D_L^2)^{1/2}=L^{\theta_{\mathrm{rms}}+o(1)}$, then
\begin{equation}\label{eq:theta-upper}
  \theta_{\mathrm{rms}}\leq\frac{d-1}{2}.
\end{equation}
All constants are independent of $\beta$, and the conclusion also holds for
ground-state energies. These hypotheses do not imply a matching lower bound.
\end{theorem}

\section{Normalization and random versus averaged quantities}

By \eqref{eq:counts} and \eqref{eq:Delta-def}, for either sample or expected
free energies,
\begin{equation}\label{eq:normalization}
  L(f_L^{\mathrm{fix}}-f_L^{\mathrm{free}})
  =\frac{\Delta F_L}{L^{d-1}}
  =2d\,\frac{\Delta F_L}{b_L}.
\end{equation}
Thus a total boundary correction of order $L^{d-1}$ has an order-one surface
free-energy density. The normalization by $L^{d-1}$ differs from the
normalization by the number of boundary bonds by the factor $2d$.

It is useful to distinguish the following four quantities:
\begin{center}
\small
\begin{tabularx}{0.96\textwidth}{>{\raggedright\arraybackslash}p{0.25\textwidth}
  >{\raggedright\arraybackslash}p{0.23\textwidth}X}
\toprule
Object & Deterministic? & Possible scale and interpretation \\
\midrule
$\E\Delta F_L$ & Yes, by definition & Disorder-averaged free-to-fixed boundary free-energy difference; it may be of order $L^{d-1}$. \\
$\Delta F_L-\E\Delta F_L$ & No & Sample fluctuation, containing local boundary noise and possibly bulk-mediated response. \\
$D_L=F_L^{\mathrm{tw}}-F_L^{\mathrm{untw}}$ & No & Seam-flip difference in the same disorder realization; a domain-wall observable for which the mean cancels by symmetry. \\
$\Var D_L$ or $\E\abs{D_L}$ & Yes & Different definitions of ``typical size''; they need not have identical exponents without moment control. \\
\bottomrule
\end{tabularx}
\end{center}

When it exists, define the root-mean-square stiffness exponent
$\theta_{\mathrm{rms}}$ by
\begin{equation}\label{eq:theta-definition}
  \norm{D_L}_2=(\E D_L^2)^{1/2}
  =L^{\theta_{\mathrm{rms}}+o(1)}.
\end{equation}
Other notions, such as a median exponent or convergence of a rescaled law,
must be stated separately. Under the root-mean-square definition,
\begin{equation}\label{eq:theta-density}
  \norm{\frac{D_L}{b_L}}_2
  =L^{\theta_{\mathrm{rms}}-(d-1)+o(1)},
  \qquad
  \norm{L(f_L^{\mathrm{tw}}-f_L^{\mathrm{untw}})}_2
  =L^{\theta_{\mathrm{rms}}-d+1+o(1)}.
\end{equation}
A mean boundary contribution $\tau b_L$ and centered fluctuations of order
$L^{\theta_{\mathrm{rms}}}$ may coexist. If
$\theta_{\mathrm{rms}}<d-1$, the surface free-energy density can converge to
$\tau$ while the unnormalized correction still fluctuates.

There is a useful uniform-integrability bound. For
$W_L=\sum_{e\in\BL}\abs{K_e}$, the finite-volume comparison proved below gives
$\abs{\Delta F_L}\leq W_L$. If $p\geq1$ and $\E\abs{K}^p<\infty$, Jensen's
inequality gives
\begin{equation}\label{eq:UI}
  \E\left\lvert\frac{\Delta F_L}{L^{d-1}}\right\rvert^p
  \leq(2d)^p\E\abs{K}^p.
\end{equation}
For merely integrable $K$, de la Vall\'{e}e Poussin's criterion applied to the
boundary averages gives uniform integrability. Hence passage from convergence
in probability to convergence of means is legitimate whenever the stated
mode and moment assumptions supply it; it is not automatic from an
$O(L^{d-1})$ estimate alone.

\section{Boundary interpolation identities}
\label{sec:interpolation}

The calculations below use the finite-volume Gaussian interpolation and
integration-by-parts framework standard in spin-glass theory; compare
\cite{GuerraToninelli2002}. All identities needed here are derived explicitly.

\subsection{The proposed interpolation}

Write $\mu=\E\abs{K_e}$ and $q_e=\sigma_x\psi_y$. The proposed boundary
coefficient is
\[
  h_e(t)=\sqrt t\,K_e+(1-\sqrt t)\mu
        =\mu+\sqrt t(K_e-\mu).
\]
Therefore its endpoints are
\begin{align*}
  \widetilde H_0&=H_L^{\mathrm{free}}-\mu\sum_{e\in\BL}q_e,\\
  \widetilde H_1&=H_L^{\mathrm{free}}-\sum_{e\in\BL}K_eq_e.
\end{align*}
The first endpoint is a deterministic boundary field of strength $\mu$, not the free
Hamiltonian. Consequently this path does not compute
$F_L^{\mathrm{fix}}-F_L^{\mathrm{free}}$.

Let $\widetilde\Phi_L(t)=\E\log\widetilde Z_{L,t}$ and
$m_e(t)=\langle q_e\rangle_t$. For $t>0$, finite-volume differentiation gives
\begin{equation}\label{eq:bad-raw}
  \widetilde\Phi_L'(t)
  =\frac{\beta}{2\sqrt t}
    \sum_{e\in\BL}\E\bigl[(K_e-\mu)m_e(t)\bigr].
\end{equation}
If the $K_e$ are i.i.d. $N(0,\kappa^2)$, Gaussian integration by parts and
$\partial_{K_e}m_e=\beta\sqrt t(1-m_e^2)$ yield
\begin{equation}\label{eq:bad-derivative}
  \widetilde\Phi_L'(t)
  =\frac{\beta^2\kappa^2}{2}
      \sum_{e\in\BL}\E(1-m_e(t)^2)
   -\frac{\beta\mu}{2\sqrt t}
      \sum_{e\in\BL}\E m_e(t).
\end{equation}
In general $m_e(0)\neq0$ because $\widetilde H_0$ already contains a boundary
field. The second term is then of order $t^{-1/2}$. It is integrable
on $(0,1)$, but the derivative need not have a finite right limit or a sign.
Thus the derivative may have neither a finite right limit nor a fixed sign.
More importantly, this path compares the random boundary condition with a
deterministic boundary field rather than with the free condition.

\subsection{Gibbs covariance identities}

For any $t$-independent observable $A$ and a differentiable Hamiltonian,
\begin{equation}\label{eq:gibbs-t-derivative}
  \partial_t\langle A\rangle_t
  =-\beta\Cov_t(A,\partial_tH_t).
\end{equation}
For the Gaussian interpolation defined below,
\begin{equation}\label{eq:gibbs-K-derivative}
  \partial_{K_b}\langle A\rangle_t
  =\beta\sqrt t\,\Cov_t(A,q_b),
\end{equation}
so that
\[
  \partial_{K_b}m_e(t)
  =\beta\sqrt t\bigl(\langle q_eq_b\rangle_t-m_e(t)m_b(t)\bigr).
\]
The cross-covariances must be retained for correlated Gaussian boundary
disorder. They disappear in the i.i.d. formula only because the disorder
covariance is diagonal.

\subsection{Gaussian boundary interpolation}

Define
\begin{equation}\label{eq:correct-path}
  H_{L,t}=H_L^{\mathrm{free}}
       -\sqrt t\sum_{e\in\BL}K_eq_e,
  \qquad t\in[0,1],
\end{equation}
and $\Phi_L(t)=\E\log Z_{L,t}$. Now $H_{L,0}=H_L^{\mathrm{free}}$ and
$H_{L,1}=H_L^{\mathrm{fix}}$. For $t>0$,
\begin{equation}\label{eq:correct-raw}
  \Phi_L'(t)=\frac{\beta}{2\sqrt t}
      \sum_{e\in\BL}\E[K_em_e(t)].
\end{equation}
If $K$ is centered Gaussian with covariance matrix $\Gamma$, integration by parts
and \eqref{eq:gibbs-K-derivative} give the identity
\begin{equation}\label{eq:correlated-gaussian}
  \Phi_L'(t)=\frac{\beta^2}{2}
  \sum_{e,b\in\BL}\Gamma_{eb}\,
  \E\Cov_t(q_e,q_b).
\end{equation}
Both $\Gamma$ and the Gibbs covariance matrix are positive semidefinite, so the
right side is one half of $\beta^2$ times the trace of a product of two
positive semidefinite matrices and is nonnegative. For i.i.d. variance
$\kappa^2$ this reduces to
\begin{equation}\label{eq:gaussian-diagonal-derivative}
  \Phi_L'(t)=\frac{\beta^2\kappa^2}{2}
  \sum_{e\in\BL}\E\bigl[1-m_e(t)^2\bigr].
\end{equation}

All limit interchanges here are finite-volume statements with explicit
domination. For $t>0$, the raw derivative satisfies the integrable bound
\[
  \abs{\Phi_L'(t)}
  \leq\frac{\beta}{2\sqrt t}\sum_{e\in\BL}\E\abs{K_e}.
\]
At the endpoint,
\begin{equation}\label{eq:continuity-zero}
  \abs{\log Z_{L,t}-\log Z_{L,0}}
  \leq\beta\sqrt t\sum_{e\in\BL}\abs{K_e},
\end{equation}
which proves continuity in $L^1$ at zero. For correlated Gaussian disorder,
the right side of \eqref{eq:correlated-gaussian} is bounded at fixed $L$ by a
constant depending only on $\beta$ and $\Gamma$. In the i.i.d. case,
\eqref{eq:gaussian-diagonal-derivative} is bounded by
$\beta^2\kappa^2b_L/2$. Thus $\Phi_L$ is absolutely continuous on $[0,1]$,
and integration is justified. In the i.i.d. case this gives
\begin{equation}\label{eq:gaussian-Delta}
  \E\Delta F_L
  =-\frac{\beta\kappa^2}{2}\sum_{e\in\BL}
    \int_0^1\E\bigl[1-m_e(t)^2\bigr]\dd t.
\end{equation}
In particular,
\begin{equation}\label{eq:gaussian-bound}
  -\frac{\beta\kappa^2}{2}b_L
  \leq\E\Delta F_L\leq0.
\end{equation}
If $q_e^{(1)}$ and $q_e^{(2)}$ denote the observables in two conditionally
independent Gibbs replicas, then
$m_e^2=\langle q_e^{(1)}q_e^{(2)}\rangle_t$ gives the two-replica
representation. Formula \eqref{eq:gaussian-Delta} concerns the
disorder average; convergence of its boundary average is a separate question.

\subsection{Independent-copy identity for non-Gaussian disorder}

No Gaussian integration by parts is needed for a finite-variance centered
boundary-coupling law. Let $K_e'$ be an independent copy, condition on all other
coordinates, and regard $m_e(x)$ as a function of the $e$th coupling. Then
\begin{align}
  \E[K_em_e(K_e)]
  &=\frac12\E\bigl[(K_e-K_e')
                  (m_e(K_e)-m_e(K_e'))\bigr] \notag\\
  &=\frac{\beta\sqrt t}{2}\E\!\left[(K_e-K_e')^2
       \int_0^1\bigl(1-m_{e,u}^2\bigr)\dd u\right],
       \label{eq:independent-copy}
\end{align}
where $m_{e,u}$ is evaluated with the $e$th coupling
$K_e'+u(K_e-K_e')$. Substitution in \eqref{eq:correct-raw} yields
\begin{equation}\label{eq:nongaussian-derivative}
  \Phi_L'(t)=\frac{\beta^2}{4}\sum_{e\in\BL}
  \E\!\left[(K_e-K_e')^2
       \int_0^1(1-m_{e,u}^2)\dd u\right].
\end{equation}
Thus $0\leq\Phi_L'(t)\leq\beta^2v_Kb_L/2$ and
$-\beta v_Kb_L/2\leq\E\Delta F_L\leq0$.

If only $\E\abs K<\infty$ is assumed, use the nonsingular linear path
\begin{equation}\label{eq:linear-path}
  H_{L,s}=H_L^{\mathrm{free}}-s\sum_{e\in\BL}K_eq_e,
  \qquad
  \frac{\dd}{\dd s}\E\log Z_{L,s}
  =\beta\sum_{e\in\BL}\E[K_e\langle q_e\rangle_s].
\end{equation}
The derivative is dominated by $\beta b_L\E\abs K$. This is the path used in
the surface-limit proof.

\subsection{Pointwise comparison and temperature-uniform concentration}

Let $V_L^\partial(\sigma)=\sum_{e\in\BL}K_eq_e(\sigma)$. The free Gibbs measure is
invariant under the global flip $\sigma\mapsto-\sigma$. Therefore
\begin{equation}\label{eq:ratio-cosh}
  \frac{Z_L^{\mathrm{fix}}}{Z_L^{\mathrm{free}}}
  =\left\langle e^{\beta V_L^\partial}\right\rangle_{\mathrm{free}}
  =\left\langle\cosh(\beta V_L^\partial)\right\rangle_{\mathrm{free}}.
\end{equation}
Since $1\leq\cosh(\beta V_L^\partial)\leq
e^{\beta\sum_e\abs{K_e}}$,
\begin{equation}\label{eq:sample-bound}
  -\sum_{e\in\BL}\abs{K_e}\leq\Delta F_L\leq0.
\end{equation}
Without global spin-flip symmetry only the absolute comparison
$\abs{\Delta F_L}\leq\sum_e\abs{K_e}$ remains valid.

If all internal and boundary bonds are independent $N(0,\kappa^2)$, then
\begin{align*}
  \abs{\partial_{J_a}\Delta F_L}&\leq2,
       &&a\in\EL,\\
  \abs{\partial_{K_e}\Delta F_L}&\leq1,
       &&e\in\BL.
\end{align*}
The Gaussian Poincar\'{e} inequality and Gaussian concentration
\cite{Ledoux2001} consequently give, for every $\beta$,
\begin{align}
  \Var(\Delta F_L)
  &\leq\kappa^2(4\abs{\EL}+b_L)
   =\kappa^2d(4L-2)L^{d-1},
   \label{eq:generic-var}\\
  \Pp\bigl(\abs{\Delta F_L-\E\Delta F_L}\geq u\bigr)
  &\leq2\exp\!\left[
     -\frac{u^2}{2\kappa^2(4\abs{\EL}+b_L)}\right].
   \label{eq:generic-tail}
\end{align}
Hence
\begin{equation}\label{eq:generic-surface-var}
  \Var\!\left(\frac{\Delta F_L}{L^{d-1}}\right)
  \leq\frac{\kappa^2d(4L-2)}{L^{d-1}}
  =O(L^{2-d}).
\end{equation}
For $d\geq3$, if the normalized means
$\E\Delta F_L/L^{d-1}$ converge, the sample surface variables converge to the
same limit in $L^2$ and almost surely; the tail in \eqref{eq:generic-tail} is
summable. For $d=2$ this estimate does not prove self-averaging. In every
dimension, concentration alone does not imply convergence of the expectations.

\section{High-temperature surface limit}

\subsection{Dobrushin comparison and the half-space state}

For an observable $A$ on a finite spin space and a site $i$, define its
single-site oscillation by
\[
  \delta_i(A)=\sup\{\abs{A(\sigma)-A(\sigma')}:
  \sigma_j=\sigma_j'\text{ for every }j\neq i\}.
\]
We use the following standard form of Dobrushin's comparison theorem
\cite{Georgii2011}.

\begin{lemma}[Dobrushin comparison estimate]\label{lem:dobrushin}
Let $\mu$ and $\nu$ be Gibbs measures on the same finite site set. Suppose
their single-site specifications have a common influence matrix $\mathcal C$
with nonnegative entries and row sums at most $\alpha<1$. Assume that
$\mathcal C_{ij}=0$ unless $i$ and $j$ are neighbors in the underlying graph. Let
$\mathcal D=(I-\mathcal C)^{-1}$, and let $b_j\leq1$ be a uniform bound on the
total-variation distance between the two single-site specifications at $j$.
Then
\begin{equation}\label{eq:dob-general}
  \abs{\mu(A)-\nu(A)}
  \leq\sum_i\delta_i(A)\sum_j\mathcal D_{ij}b_j.
\end{equation}
If, for an integer $r\geq1$, the specifications agree on the closed
radius-$(r-1)$ neighborhood of $\operatorname{supp}A$, then
\begin{equation}\label{eq:dob-comparison}
  \abs{\mu(A)-\nu(A)}
  \leq\frac{\alpha^r}{1-\alpha}
       \sum_{i\in\operatorname{supp}A}\delta_i(A).
\end{equation}
\end{lemma}

\begin{proof}
The first inequality is Dobrushin's comparison theorem. For the second, write
$\mathcal D=\sum_{n\geq0}\mathcal C^n$. Since $b_j=0$ at distance less than
$r$ from $\operatorname{supp}A$ and the row sums of $\mathcal C$ are at most
$\alpha$,
\[
  \sum_j\mathcal D_{ij}b_j
  \leq\sum_{n\geq r}(\mathcal C^n\mathbf1)_i
  \leq\frac{\alpha^r}{1-\alpha}.
\]
Substitution in \eqref{eq:dob-general} proves the claim.
\end{proof}

In regime (D) we may take
\[
  \mathcal C_{xz}=\tanh(\beta J_*)\mathbf1_{\{x\sim z\}},
  \qquad \sup_x\sum_z\mathcal C_{xz}\leq\alpha<1.
\]
Boundary interactions with prescribed exterior spins are one-site fields and
therefore do not enter $\mathcal C$. In particular, the estimate is uniform in
the interpolation parameter, the disorder realization, and the exterior
spins.

We next construct the half-space state appearing in \eqref{eq:g-half}. Take
any increasing finite exhaustion $(H_n)$ of $\mathbb H$ and impose arbitrary
exterior conditions. For $m>n$, condition the Gibbs measure on $H_m$ on the
spins in $H_m\setminus H_n$. Its conditional distribution and the Gibbs
measure on $H_n$ are then defined on the same site set. For a fixed local
observable, their specifications agree in a ball whose radius tends to
infinity with $n$, uniformly in the conditioned spins. By
\cref{lem:dobrushin}, the corresponding expectations form a uniformly Cauchy
sequence after averaging over those spins. The limit defines a Gibbs state on
$\mathbb H$, is independent of the exhaustion and exterior conditions, and is
measurable in the disorder as a pointwise limit of finite-volume expectations.
The same estimate proves uniqueness.

\subsection{Proof of the surface-limit theorem}

We now prove \cref{thm:surface}. For later use, define the inner vertex
boundary and the depth of an internal edge by
\[
  \partial_{\mathrm{in}}\Lambda
  =\{x\in\Lambda:\operatorname{dist}(x,\Lambda^c)=1\},
  \qquad
  r(a)=\operatorname{dist}(\{x,z\},\partial_{\mathrm{in}}\Lambda),
  \quad a=\{x,z\}.
\]

Pointwise finite-volume differentiation of the linear path
\eqref{eq:linear-path} gives
\begin{equation}\label{eq:linear-integrated}
  \E\Delta F_L
  =-\int_0^1\sum_{e\in\BL}
     \E[K_e\langle q_e\rangle_{L,s}]\dd s.
\end{equation}
The integrand is bounded by $J_*b_L$, so Fubini's theorem applies.

Fix the face $x_1=1$, and consider a boundary edge $e=(x,x-\mathbf e_1)$
whose tangential distance from the edges of that face is at least $R$.
Use the same disorder in the cube and the half-space. Condition the half-space
state on the spins in $\mathbb H\setminus\LambdaL$. The resulting conditional
measure and the cube measure are Gibbs measures on the same site set
$\LambdaL$. Their specifications agree on the closed radius-$(R-1)$
neighborhood of $x$: on the target face they have the same interpolated
boundary field, while discrepancies on the other faces occur at distance at
least $R$. Applying \cref{lem:dobrushin} conditionally and then averaging over
the conditioned spins gives, uniformly in $s$,
\[
  \abs{\E[K_e\langle q_e\rangle_{L,s}]-g_{\mathbb H}(s)}
  \leq\frac{4J_*}{1-\alpha}\alpha^R.
\]
Here we used the uniform bound $\sum_i\delta_i(q_e)\leq4$. The same argument
applies to the other faces. By lattice symmetries all $2d$ half-space limits
agree.

For an integer $1\leq R\leq L/2$, the fraction of boundary edges excluded by
the tangential-distance condition is at most $2(d-1)R/L$. Since each boundary
integrand has absolute value at most $J_*$, averaging over the good and bad
edges yields
\begin{equation}\label{eq:face-error}
  \sup_{s\in[0,1]}
  \left\lvert\frac1{b_L}\sum_{e\in\BL}
    \E[K_e\langle q_e\rangle_{L,s}]-g_{\mathbb H}(s)\right\rvert
  \leq4(d-1)J_*\frac RL
      +\frac{4J_*}{1-\alpha}\alpha^R.
\end{equation}
For $0<\alpha<1$, choose $R$ proportional to $\log L$; when $\alpha=0$, take
$R=1$. The right side tends to zero uniformly in $s$. Dominated convergence
in \eqref{eq:linear-integrated} proves
$\E\Delta F_L/b_L\to\tau(\beta)$ with \eqref{eq:tau-half}. The independent-copy
identity along the linear path gives, in every finite half-space exhaustion,
\[
  \E[K_e m_{e,s}(K_e)]
  =\frac{\beta s}{2}\E\!\left[(K_e-K_e')^2
       \int_0^1(1-m_{e,s,u}^2)\dd u\right],
\]
where $m_{e,s,u}$ is evaluated with the $e$th coupling
$K_e'+u(K_e-K_e')$.
The Dobrushin estimate and bounded convergence allow passage to the half-space
limit. Consequently,
\[
  0\leq g_{\mathbb H}(s)\leq\beta s v_K,
\]
which proves $-\beta v_K/2\leq\tau\leq0$.
In particular, the same explicit right side in \eqref{eq:face-error} bounds
$\abs{\E\Delta F_L/b_L-\tau(\beta)}$.

It remains to control the centered random correction. If an internal bond
$J_a$ at depth $r(a)=r$ is resampled, integrate
\[
  \sigma_a:=\sigma_x\sigma_z\quad(a=\{x,z\}),
  \qquad
  \partial_{J_a}\Delta F_L
  =-\langle\sigma_a\rangle_{\mathrm{fix}}
   +\langle\sigma_a\rangle_{\mathrm{free}}
\]
along the resampling segment. For $r\geq1$, estimate
\eqref{eq:dob-comparison} is uniform on that segment. For $r=0$, the general
estimate \eqref{eq:dob-general} gives the same bound with $\alpha^0=1$.
Therefore the coordinate oscillation is at most
\begin{equation}\label{eq:internal-sensitivity}
  c_a\leq\frac{8J_*}{1-\alpha}\alpha^r.
\end{equation}
A boundary-bond oscillation is at most $2J_*$. The vertices at a fixed depth
form a subset of the vertex boundary of a smaller cube and hence number at
most $b_L$; since each vertex is incident to at most $2d$ internal edges,
there are at most $2db_L$ internal edges at a fixed depth. Summing the squared
oscillations over the depth layers gives
\begin{equation}\label{eq:Q-bound}
  \sum_i c_i^2
  \leq b_L\left[(2J_*)^2+
    \frac{32d(2J_*)^2}{(1-\alpha)^2(1-\alpha^2)}\right]
  =C_*b_L.
\end{equation}
McDiarmid's bounded-differences inequality \cite{BLM2013} and
\eqref{eq:Q-bound} give
\eqref{eq:surface-mcdiarmid}. With
$u=\varepsilon b_L$, its right side is
$2\exp(-2\varepsilon^2b_L/C_*)$, summable for integer cubes when $d\geq2$.
Borel--Cantelli gives almost-sure convergence to the mean. Integrating the
tail gives convergence in $L^p$ for every $1\leq p<\infty$.

The same argument applies to a rectangular box
$\Lambda_{\mathbf L}=\prod_{i=1}^d\{1,\ldots,L_i\}$. With
$V=\prod_iL_i$, $\ell=\min_iL_i$, and
\[
  b_{\mathbf L}=2V\sum_{i=1}^dL_i^{-1},
\]
the fraction of bad edges on any face is at most
$2(d-1)R/\ell$. Thus the right side of \eqref{eq:face-error} remains valid
with $L$ replaced by $\ell$. At each depth the number of internal edges is at
most $2db_{\mathbf L}$, so \eqref{eq:Q-bound} holds with $b_L$ replaced by
$b_{\mathbf L}$. Taking $R$ proportional to $\log\ell$ proves convergence of
the mean and convergence in every finite $L^p$ as $\ell\to\infty$.
Almost-sure convergence follows whenever
$\sum_n\exp(-c b_{\mathbf L_n})<\infty$ for every $c>0$. This completes the
proof of \cref{thm:surface}.

\subsection{Dependence on the van Hove sequence}
\label{sec:counterexamples}

We now prove the first part of \cref{thm:impossibility} using a standard
i.i.d. symmetric one-dimensional model. Let $K_i$ be i.i.d. Rademacher variables
taking values $\pm K$, put $x=\beta K>0$, and fix all exterior spins plus.
For one interval of $n$ sites, the high-temperature expansion gives exactly
\begin{equation}\label{eq:chain-ratio}
  \log\frac{Z_n^+}{Z_n^{\mathrm{free}}}
  =X_0+X_n+
  \log\!\left(1+\prod_{i=0}^{n}T_i\right),
  \quad
  X_i=\log\cosh(\beta K_i),\quad
  T_i=\tanh(\beta K_i).
\end{equation}
For other endpoint spins the product receives their relative sign. Since
$\abs{T_i}=\tanh x<1$, the final term tends to zero uniformly. The expected
difference of log-partition functions per boundary edge therefore tends to
\begin{equation}\label{eq:mu-interval}
  \mu_\mathrm{int}=\log\cosh x.
\end{equation}

Construct a nested exhaustion recursively. Choose $N_1\geq1$ and let
$I_k=[-N_k,N_k]\cap\mathbb Z$. Given $I_k$, set
$M_k=\lfloor\sqrt{\abs{I_k}}\rfloor$ and add, to the right of $I_k$, $M_k$
isolated sites spaced three lattice units apart; call the union $A_k$. Choose
$N_{k+1}>N_k$ so that $I_{k+1}$ contains $A_k$. Then
\[
  I_1\subset A_1\subset I_2\subset A_2\subset\cdots
\]
is a nested exhaustion. Moreover,
\[
  \frac{\abs{B(I_k)}}{\abs{I_k}}\longrightarrow0,
  \qquad
  \frac{\abs{B(A_k)}}{\abs{A_k}}
  =\frac{2+2M_k}{\abs{I_k}+M_k}\longrightarrow0,
\]
so it is a van Hove sequence. The same estimate, with a constant depending on
any fixed $r$, verifies the van Hove condition stated in \cref{sec:model}.

\begin{figure}[htbp]
	\centering
	\begin{minipage}{0.48\textwidth}
		\centering
		\includegraphics[width=\textwidth]{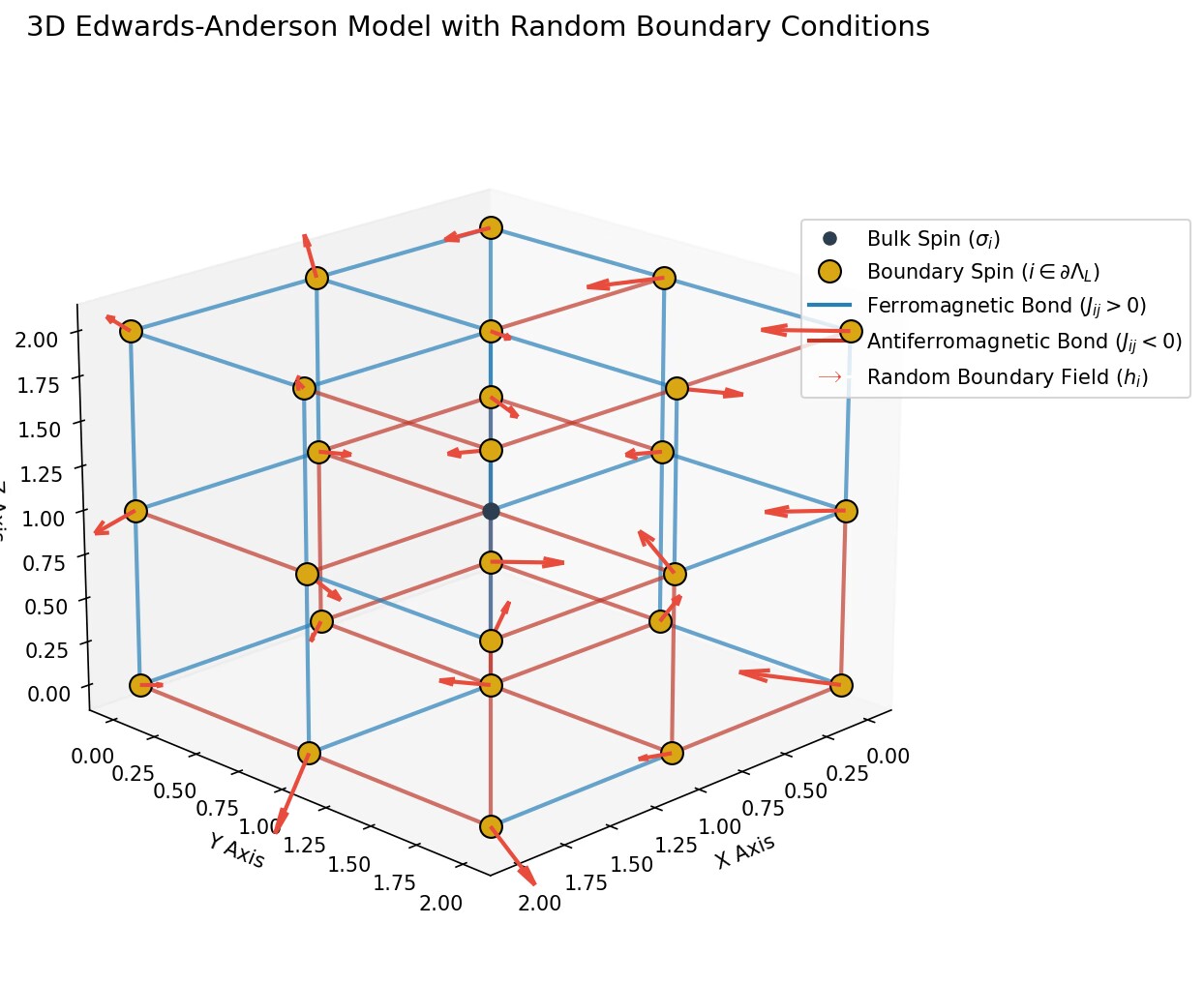}
		\caption{Preliminary topological preview of the one-dimensional nested van Hove sequence, outlining the basic geometric spacing between the intervals and isolated boundary sites.}
		\label{fig:van_hove_preview}
	\end{minipage}
	\hfill
	\begin{minipage}{0.48\textwidth}
		\centering
		\includegraphics[width=\textwidth]{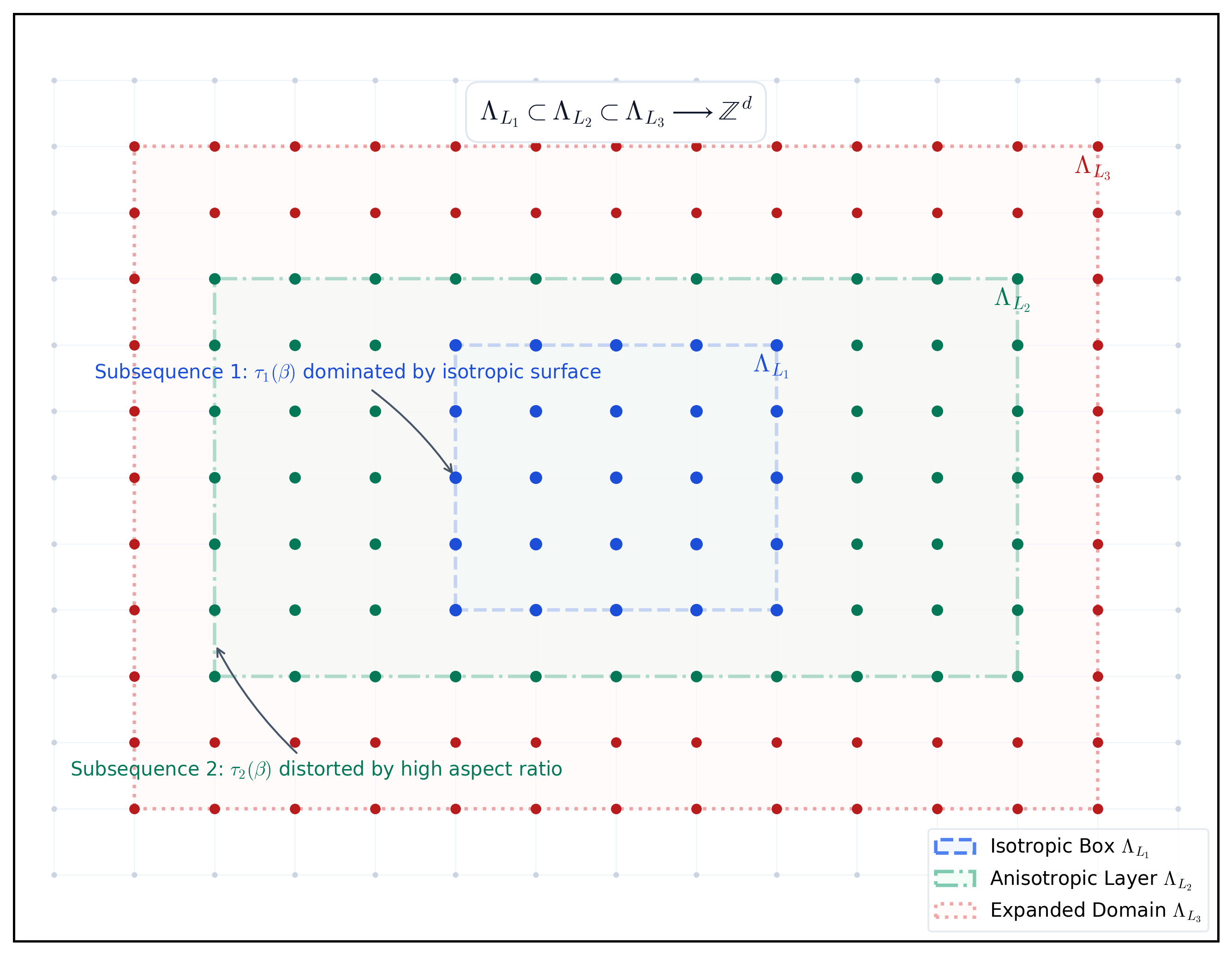}
		\caption{High-resolution diagram of the counterexample exhaustion path on $\mathbb{Z}$, demonstrating how alternating geometric densities trigger distinct subsequential limits for the expected surface free-energy term.}
		\label{fig:van_hove_highres}
	\end{minipage}
\end{figure}

For one isolated spin the partition-function ratio is
$\cosh(\beta(K_1+K_2))$. Since $K_1+K_2$ equals $0$ with probability $1/2$
and $\pm2K$ with total probability $1/2$, its expected log excess per one of
its two boundary edges is
\begin{equation}\label{eq:mu-isolated}
  \mu_\mathrm{iso}=\frac14\log\cosh(2x).
\end{equation}
The isolated components dominate the boundary of $A_k$, while the long
interval contributes only two boundary edges. Hence the expected log excess
per edge tends to $\mu_\mathrm{int}$ along $I_k$ and to $\mu_\mathrm{iso}$
along $A_k$. These constants differ because
\[
  \cosh^4x-\cosh(2x)=(\cosh^2x-1)^2>0.
\]
Multiplication by $-1/\beta$ gives two distinct free-energy surface limits.
The example can be chosen at arbitrarily small $x$, so it persists inside the
one-dimensional Dobrushin regime. This proves that arbitrary van Hove
independence is false even for bounded i.i.d. symmetric finite-range disorder.

The mechanism is microscopic boundary geometry, not frustration. In a factorized model
with no internal interactions, if $k_x$ boundary edges meet a boundary vertex,
\[
  \Delta F_\Lambda=-\frac1\beta\sum_{x\in\Lambda}
  \log\cosh\!\left(\beta\sum_{e=(x,y)\in B(\Lambda)}K_e\psi_y\right).
\]
Flat faces have $k_x=1$ almost everywhere; a microscopic staircase can have a
positive density of $k_x=2$. Division by the number of boundary edges therefore
produces different constants. Regularity assumptions on the boundary are
therefore essential.

\subsection{Failure of sample convergence in dimension one}

Take centered intervals $[-n,n]$ in a common i.i.d. environment and assume that
$X=\log\cosh(\beta K)$ is nondegenerate and integrable, for example by giving
$\abs K$ two bounded positive values. Formula \eqref{eq:chain-ratio} shows that
the sample correction is, up to a term tending to zero,
\[
  -\frac1\beta(X_{-n-1}+X_n).
\]
The boundary pairs for different $n$ are independent and identically
distributed. This sequence is not Cauchy in probability and hence has no
limit in probability or almost surely, although it converges in distribution
and its expectation converges. In $d=1$ the boundary contains only two bonds,
so no boundary averaging occurs. Thus the mode of convergence must be
specified.

\section{Domain-wall free-energy fluctuations}

\subsection{Proof of the seam variance bound}

We prove \cref{thm:seam}. Related interface free-energy fluctuation estimates
are developed in \cite{ArguinNewmanSteinWehr2014}. Write the non-seam disorder
as $\omega$ and the seam
vector as $y=(y_e)_{e\in S_L}$. Let
$F_{L,\beta}(\omega,y)$ denote the free energy of the untwisted system, and let
$T$ flip every coordinate of $y$. Put
\begin{equation}\label{eq:D-seam-def}
  D_L(\omega,y)=F_{L,\beta}(\omega,Ty)-F_{L,\beta}(\omega,y).
\end{equation}
Because the joint seam law is invariant under $T$,
\[
  D_L(\omega,Ty)=-D_L(\omega,y),
  \qquad \E_y[D_L(\omega,y)\mid\omega]=0.
\]
The free energy is one-Lipschitz in each bond: its derivative, where it exists,
is a Gibbs expectation of a $\{-1,+1\}$ observable. Replacing one $y_e$ by an
independent copy $y_e'$ changes each of the two free energies in
\eqref{eq:D-seam-def} by at most $\abs{y_e-y_e'}$. Hence
\[
  \abs{D_L(y)-D_L(y^{(e)})}\leq2\abs{y_e-y_e'}.
\]
Here $y^{(e)}$ denotes the vector obtained from $y$ by replacing $y_e$ with
$y_e'$.
The conditional Efron--Stein inequality \cite{BLM2013} gives
\begin{align}
  \E D_L^2
  &=\E\Var_y(D_L\mid\omega)\notag\\
  &\leq\frac12\sum_{e\in S_L}
     \E\bigl[4(y_e-y_e')^2\bigr]
   =4v\abs{S_L}.
   \label{eq:seam-efron}
\end{align}
The Hamiltonians differ by at most
$2\sum_{e\in S_L}\abs{y_e}$ uniformly in $\sigma$, proving the deterministic
bound in \eqref{eq:seam-summary}.

For Gaussian seam bonds, $\abs{\partial_{y_e}D_L}\leq2$. Since
$\E_y[D_L\mid\omega]=0$, conditional Gaussian concentration gives
\eqref{eq:seam-tail-summary}; integrating over $\omega$ preserves the same
bound.
If $\abs{S_L}\asymp L^{d-1}$, then
\begin{equation}\label{eq:seam-density-vanish}
  \E\left(\frac{D_L}{L^{d-1}}\right)^2
  =O(L^{-(d-1)}).
\end{equation}
When $d\geq2$, the domain-wall surface free-energy density therefore tends to
zero in $L^2$. In the Gaussian case it also converges almost surely for
$d\geq2$, and the tail estimate gives
$D_L=O(L^{(d-1)/2+\varepsilon})$ almost surely for every $\varepsilon>0$.
The root-mean-square definition \eqref{eq:theta-definition} then yields
\eqref{eq:theta-upper}.

These constants contain no $\beta$. The ground-state energy
$\min_\sigma H(\sigma)$ has the same one-bond Lipschitz property, so the proof
holds verbatim at $\beta=\infty$ and uniformly for a sequence $\beta_L$.
This does not imply that the actual finite-temperature and zero-temperature
exponents coincide.

For periodic versus antiperiodic conditions, $T$ flips the couplings on one
wrapping seam. For two exterior patterns, $T$ flips precisely the boundary
bonds on which their effective signs differ. In both cases the seam-flip mean
is zero under symmetric disorder. The free-to-fixed boundary free-energy
difference is not of this antisymmetric form and does not enjoy the seam-only
variance bound.

\subsection{A variance bound for the free-to-fixed comparison}

Combining \eqref{eq:sample-bound} and \eqref{eq:generic-var} gives, for the
centered free-to-fixed correction under i.i.d. Gaussian disorder,
\begin{equation}\label{eq:generic-theta}
  \norm{\Delta F_L-\E\Delta F_L}_{2}
  =O\!\left(\min\{L^{d/2},L^{d-1}\}\right).
\end{equation}
Thus a root-mean-square power law for this centered variable can only imply
$\theta_{\mathrm{rms}}\leq\min\{d/2,d-1\}$. This is weaker than
\eqref{eq:theta-upper} because internal disorder contributes to the general
Poincar\'{e} bound and there is no conditional sign-flip cancellation.

\subsection{Boundary fluctuations in a decoupled-chain model}

Assume $d\geq2$ and consider an anisotropic finite-range model consisting of
$L^{d-1}$ independent zero-field chains, with one random boundary bond $K_a$
at one end of each chain. Use the first $L^{d-1}$ variables of a single i.i.d.
sequence $(K_a)_{a\geq1}$ as $L$ varies. The endpoint marginal of each free
chain is uniform, so exactly
\begin{equation}\label{eq:chain-CLT}
  \Delta F_L=-\frac1\beta\sum_{a=1}^{L^{d-1}}
                  \log\cosh(\beta K_a).
\end{equation}
If $\log\cosh(\beta K)$ has nonzero finite variance, the law of large numbers
and central limit theorem give
\begin{align*}
  \frac{\Delta F_L}{L^{d-1}}
  &\xrightarrow[L\to\infty]{\mathrm{a.s.},\,L^2}
    -\frac1\beta\E\log\cosh(\beta K),\\
  \frac{\Delta F_L-\E\Delta F_L}{L^{(d-1)/2}}
  &\Longrightarrow N\!\left(0,
    \frac1{\beta^2}\Var(\log\cosh(\beta K))\right).
\end{align*}
These fluctuations arise from independent local boundary contributions rather
than from a macroscopic domain wall. A stiffness observable should therefore
use a comparison in which these local contributions cancel.

\subsection{One-dimensional periodic versus antiperiodic examples}

For a one-dimensional ring, define
$D_N=F_N^{\mathrm{aper}}-F_N^{\mathrm{per}}$. The high-temperature expansion,
with $Q_N=\prod_{i=1}^N\tanh(\beta J_i)$, gives
\begin{equation}\label{eq:ring-twist}
  D_N=\frac1\beta\log\frac{1+Q_N}{1-Q_N}
      =\frac{2}{\beta}\operatorname{artanh}Q_N.
\end{equation}
At fixed finite $\beta$, if
$\E\abs{\log\abs{\tanh(\beta J)}}<\infty$ and
$\E\log\abs{\tanh(\beta J)}<0$, the strong law of large numbers shows that
$\abs{D_N}$ decays exponentially almost surely. If $\Pp(J=0)>0$, then $Q_N$
eventually vanishes and the same conclusion holds. At zero temperature the
twist energy is
\begin{equation}\label{eq:ring-ground}
  \abs{D_N}=2\min_{1\leq i\leq N}\abs{J_i}.
\end{equation}
For Rademacher bonds this is constant and $\theta_{\mathrm{rms}}=0$. For a symmetric Laplace
law with unit rate, the minimum is exponential with rate $N$, so
$(\E D_N^2)^{1/2}=2\sqrt2/N$ and $\theta_{\mathrm{rms}}=-1$. More generally,
if $\Pp(\abs J\leq x)\sim cx^a$ near zero, standard extreme-value theory
\cite{Resnick1987} gives the
exponent $-1/a$ in probability; the same root-mean-square exponent follows
under a corresponding uniform-integrability condition. Symmetry and finite
moments alone therefore supply no universal lower bound, exponent, or limiting
law.

\section{Scope and limitations}

\subsection{Comparison of methods}

\cref{tab:methods} records the precise output of each method and the additional
input required to obtain a full surface limit.

\begin{table}[ht]
\caption{Rigorous consequences of the methods used in the paper.}
\label{tab:methods}
\small
\renewcommand{\arraystretch}{1.08}
\begin{tabularx}{\textwidth}{@{}>{\raggedright\arraybackslash}p{0.20\textwidth}
 >{\raggedright\arraybackslash}p{0.35\textwidth}X@{}}
\toprule
Method & Established conclusion & Additional input required \\
\midrule
Thermodynamic-limit methods \cite{Ruelle1969}
& Existence of the bulk pressure and an $O(b_L)$ boundary comparison
& Convergence of the $b_L$-order remainder and regularity of the microscopic boundary geometry \\
Gaussian and independent-copy interpolation \cite{GuerraToninelli2002}
& Finite-volume identities for $\E\Delta F_L$, including sign and moment bounds
& Full-sequence convergence of the normalized boundary response \\
Concentration inequalities \cite{ArguinNewmanSteinWehr2014,BLM2013,Ledoux2001}
& Upper bounds for centered fluctuations and seam-order variance estimates
& Convergence of the normalized expectations and separate lower-scale information \\
Dobrushin comparison \cite{Georgii2011}
& Unique half-space Gibbs state, exponential comparison, and the regular-box surface limit
& Uniformly bounded interactions and the strict condition $\alpha<1$ \\
Metastate compactness \cite{AizenmanWehr1990,NewmanStein1997}
& Subsequential infinite-volume Gibbs states
& Uniqueness of the induced boundary response across subsequences \\
Stochastic localization \cite{Eldan2013}
& A covariance-martingale representation
& A uniform estimate for the boundary projection of the covariance process \\
\bottomrule
\end{tabularx}
\end{table}

Contour expansions yield positive surface tension in ordered ferromagnetic
phases under Peierls-type positivity assumptions \cite{DKS1992}. Those
assumptions are not available for general frustrated couplings.

\subsection{Limit passages and normalization}

The path proposed in the question does not begin at the free Hamiltonian,
whereas \eqref{eq:correct-path} and \eqref{eq:linear-path} have the required
endpoints. At finite volume, differentiation is analytic for $t>0$; bounds
\eqref{eq:continuity-zero} and \eqref{eq:gaussian-diagonal-derivative} give
continuity and absolute continuity at $t=0$. The term of order $t^{-1/2}$ in
the proposed path is integrable but need not cancel. Formula
\eqref{eq:correlated-gaussian} retains all covariance cross terms, while the
diagonal reduction applies only to i.i.d. Gaussian disorder. For non-Gaussian
couplings, \eqref{eq:independent-copy} provides the corresponding identity.

The comparison \eqref{eq:sample-bound} gives tightness and uniform
integrability, but not convergence of the normalized mean. The full-sequence
limit in \cref{thm:surface} follows from the half-space comparison. In
$d\geq3$, Gaussian concentration transfers convergence of the normalized mean
to the sample surface free-energy density; it does not establish convergence
of the mean itself. In $d=2$, the general Gaussian estimate is insufficient,
whereas Dobrushin mixing gives the surface-order concentration used in
\cref{thm:surface}.

For cubes, $b_L=2dL^{d-1}$, which accounts for the factor $2d$ in
\eqref{eq:normalization}. The counterexample in \cref{sec:counterexamples}
shows that the microscopic geometry of a general van Hove boundary can change
the normalized surface term. The high-temperature theorem keeps $\beta$ fixed
under the strict condition \eqref{eq:dob-condition}; its constants deteriorate
as $\alpha\uparrow1$. Formula \eqref{eq:gaussian-Delta} does not justify
interchanging $L\to\infty$ and $\beta\to\infty$. Only the Lipschitz seam bounds
remain uniform through $\beta=\infty$.

Finally, under symmetric boundary disorder, the quenched law of $F_L^\psi$ is
independent of a deterministic exterior pattern $\psi$. The domain-wall
observable studied here is therefore a same-disorder difference
$F_L^\psi-F_L^{\psi'}$, or a periodic versus antiperiodic seam, rather than the
free-to-fixed disorder average.

\subsection{Requirements for a stochastic-localization argument}

Following the stochastic-localization construction of \cite{Eldan2013}, a
finite-volume Gibbs measure $\mu_0$ on $\{-1,+1\}^{\LambdaL}$ would generate a
martingale of the form
\begin{equation}\label{eq:localization}
  \dd\mu_s(\sigma)
  =\mu_s(\sigma)\langle\sigma-a_s,\dd B_s\rangle,
  \qquad
  a_s=\E_{\mu_s}\sigma,
  \qquad A_s=\Cov_{\mu_s}(\sigma),
\end{equation}
where $(B_s)_{s\geq0}$ is a standard Brownian motion in
$\mathbb R^{\LambdaL}$. Then $\dd a_s=A_s\dd B_s$. To use this process for the present problem one
would need, uniformly in $L$, interpolation time, and disorder, a quantitative
bound on the boundary projection of $A_s$ and on its coupling to degrees of
freedom at macroscopic distance, strong enough both to select a unique
half-space boundary average and to control the seam-flip free-energy
difference. No such inequality is part of the proposed method, and low-temperature
Edwards--Anderson Gibbs
measures do not come with the log-concavity or spectral-independence hypothesis
that would imply it. Absent such an estimate, stochastic localization does not
establish the required limit.

\section{Open problem for the low-temperature Gaussian model}

For the standard i.i.d. Gaussian Edwards--Anderson model
\cite{EdwardsAnderson1975,NewmanStein2003} on cubic boxes at a fixed low
temperature, the Gaussian boundary interpolation proves the finite-volume
identity
\[
  \frac{\E\Delta F_L}{b_L}
  =-\frac{\beta\kappa^2}{2}\int_0^1
    \frac1{b_L}\sum_{e\in\BL}
       \E\bigl[1-\langle q_e\rangle_{L,t}^2\bigr]\dd t.
\]
The remaining surface-limit problem is to prove that the integral on the right
has a unique full-sequence limit for regular cubes, or to construct a
cube-sequence counterexample. Dobrushin comparison estimates prove this at high
temperature for bounded interactions, but it is unavailable for unbounded
Gaussian couplings and supplies no low-temperature theorem. Compactness gives
subsequences only. In $d\geq3$, the general Gaussian concentration bound shows
that a convergent mean would be the deterministic sample surface limit; in
$d=2$, that bound does not establish self-averaging.

For a periodic versus antiperiodic or boundary-sign domain wall, symmetry
forces $\E D_L=0$ at every $L$ and \cref{thm:seam} proves only the upper bound
$\theta_{\mathrm{rms}}\leq(d-1)/2$. Determining a lower scale, existence of an exponent, or a
limiting distribution requires new information about the coupled Gibbs states
and cannot be recovered from the free-to-fixed disorder-average interpolation.
The zero-temperature examples in \eqref{eq:ring-ground} show that no such
conclusion can follow from symmetry and finite moments alone.

The results above therefore clarify the model-dependent conclusions, prove a
high-temperature surface limit, and isolate the unresolved low-temperature
question. The remaining difficulty is not a formal differentiation or
concentration estimate; it is the analysis of low-temperature boundary states
and domain-wall free energies in a fully specified model.

\Addresses
\end{document}